\documentclass{article}
\usepackage[utf8]{inputenc}
\usepackage{style}
\title{Matrix Multiplication and Number On the Forehead Communication}
\author{Josh Alman\footnote{Columbia University. josh@cs.columbia.edu} \and Jarosław Błasiok\footnote{Columbia University. jb4451@columbia.edu}}

\begin{document}

\maketitle

\begin{abstract}
    Three-player Number On the Forehead communication may be thought of as a three-player Number In the Hand promise model, in which each player is given the inputs that are supposedly on the other two players' heads, and promised that they are consistent with the inputs of of the other players. The set of all allowed inputs under this promise may be thought of as an order-3 tensor. We surprisingly observe that this tensor is exactly the matrix multiplication tensor, which is widely studied in the design of fast matrix multiplication algorithms.

    Using this connection, we prove a number of results about both Number On the Forehead communication and matrix multiplication, each by using known results or techniques about the other. For example, we show how the Laser method, a key technique used to design the best matrix multiplication algorithms, can also be used to design communication protocols for a variety of problems. We also show how known lower bounds for Number On the Forehead communication can be used to bound properties of the matrix multiplication tensor such as its zeroing out subrank. Finally, we substantially generalize known methods based on slice-rank for studying communication, and show how they directly relate to the matrix multiplication exponent $\omega$.
\end{abstract}

\newpage
\section{Introduction}

%\jnote{I started writing some intro, it is quite wordy, maybe we should get rid of most of it, to get our point across much more strongly.}

Number on the forehead (NOF) communication complexity was introduced by Chandra, Furst and Lipton \cite{CFL83} as a variant of Yao's model of communication complexity. Here, we consider $k$ players, each of them having one of $k$ inputs, but instead of providing the players with their inputs, we imagine them having their inputs written on their foreheads --- in particular each player can observe all inputs \emph{except} for their own. Similar to the standard communication complexity, the players wish to exchange as little communication as possible, in order to jointly compute some function of their inputs.

As it turns out this innocuous alteration has deep mathematical consequences. For example, being able to prove lower bounds for specific NOF problems with $k = \Theta(\log^c n)$ players would imply explicit circuit complexity lower bounds \cite{HG90,BT94}.
%\jnote{I made this up, I vaguely remember Toni saying something about it at some point, and reading a statement to this effect in one of her papers... We need to check it.}
%\joshnote{I think this is Observation 17 here \url{https://www.cs.toronto.edu/~toni/Courses/CommComplexity2022/Lectures/notes-nof.pdf}}
%\jnote{Yes, thanks, that's what I had in mind.}

Moreover, even if we consider only three players, the theory of NOF communication complexity is simultaneously extremely deep, and frustratingly lacking. Its depth comes in part from a number of fascinating and perhaps surprising connections between natural questions in NOF communication complexity and well-studied problems in Ramsey Theory, extremal combinatorics and additive combinatorics. For example, already in the seminal paper \cite{CFL83}, a surprisingly efficient deterministic protocol  was provided for the NOF problem of checking whether $x+y+z = 0$ in $\bZ_N$, where $x,y,z \in \bZ_N$ are the inputs on the players' foreheads (we will call this problem $\eval_{\bZ_N}$ from now on). The protocol uses just $\Oh(\sqrt{\log N})$ communication, and makes use of Behrand's construction of a relatively dense set without three-terms arithmetic progression \cite{behrend}. On the other hand, in the same paper they showed an $\omega(1)$ lower bound for the $\eval_{\bZ_N}$ problem using the Hales–Jewett theorem --- a deep theorem in Ramsey theory. Since then, connections with the Ruzsa–Szemerédi problem, Corner-Free Sets and other objects from extremal combinatorics have been discovered, frequently leading to state-of-the-art constructions~\cite{LPS17, AS20,AMS12,LS21, CFTZ21}.

At the same time, perhaps because of these connections, many seemingly simple questions about NOF communication complexity remain unresolved, despite the fact that their number-in-hand variants are very easy to answer. For example, completely resolving the communication complexity of $\eval_{\bZ_N}$ requires a significantly better understanding of the correct quantitative dependency in the Roth's theorem (that large density subsets of $[N]$ contain a 3-terms arithmetic progression), and seems to be far out of reach of current techniques. 

Similarly, it is still open how to find an explicit $3$-party NOF problem which can be solved with $\Oh(1)$ randomized communication complexity (assuming shared randomness), but  which requires $\Omega(n)$ deterministic communication, even though this is a relatively simple property of $\textsc{Equality}$ for two-party communication, and even though, with a non-constructive argument, it is possible to show that such NOF problems exist \cite{DP08}. A natural candidate for such a problem is checking whether $x+y+z = 0$ in $\F_3^n$ (or, more generally, if $xyz = 1$ in a large power $G^n$ of some fixed group $G$, not necessarily abelian --- we will call those problems $\eval_{G^n}$), but proving a NOF communication complexity lower bound has been so far elusive.

Problems like $\eval_{G^n}$, in which for any pair of inputs to two of the players, there is exactly one input to the remaining player for which they should accept, will be called \emph{permutation problems}. These problems always have $\Oh(1)$ randomized communication complexity, so they are natural candidates to try to prove $\Omega(n)$ deterministic communication lower bounds. Indeed, a number of recent works have studied permutation problems toward this goal~\cite{BDPW10,CFTZ21}.

The construction \cite{CFL83} of an efficient protocol for $\eval_{\bZ_N}$ can be interpreted as a statement that existence of a large set $S \subset G$ without three terms arithmetic progressions in an abelian group $G$ can be used to provide an efficient protocol for $\eval_{G}$; by contrapositive, a lower bound for the communication complexity of $\eval_G$ would imply an upper bound for the size of any such set $S$. This latter statement for the group $\F_3^n$ used to be known as a \emph{cap-set conjecture} for few decades, and has famously been positively resolved in 2016 \cite{ED17,CLP17, T16} --- the largest set $S \subset \F_3^n$ without three terms arithmetic progressions has size $\Oh(c^n)$ for some $c < 3$. This leads to a natural question: can the techniques used to prove the cap-set theorem be generalized to prove a linear lower bound for the communication complexity of $\eval_{\F_3^n}$?

 Christandl et al. \cite{CFTZ21} recently identified a barrier against a direct applications of such techniques --- they used a method called \emph{combinatorial degenration} to prove that a \emph{subrank} of a specific tensor associated with the communication problem $\eval_{G^n}$ for an abelian group $G$ is near-maximal --- whereas an upper bound on this subrank would imply a desired lower-bound for communication complexity.

%\jnote{TODO: Introducing permutation problems, asymptotic communication complexity, matrix multiplication tensor and notions of subranks --- i.e. necessary terminology to state our results. Ideally not too formally.}

\subsection{Our contributions}
In this paper, we study three-player communication through the lens of \emph{promise number-in-hand problems}. For an alphabet $\Sigma$, a promise problem $(I,P)$ consists of two subsets $I \subset P \subset \Sigma \times \Sigma \times \Sigma$. In this problem, the three players are given as input $a,b,c \in \Sigma$, respectively, with the promise that $(a,b,c) \in P$, and their goal is to determine whether $(a,b,c) \in I$. For instance, usual number-in-hand communication corresponds to this model where $P = \Sigma \times \Sigma \times \Sigma$.

NOF problems can also be viewed as promise number-in-hand problems. Each player is given as an input a pair $(a_i, b_i) \in \Sigma^2$ for $i \in \{1, 2,3\}$, and their inputs are guaranteed to satisfy that $a_2 = b_1, a_3 = b_2$ and $a_1 = b_3$. (For instance, one thinks of $a_2$ and $b_1$ as the input on the head of player 3, so we are promised that players 1 and 2 see the same value there.) Write $\mmP \subset \Sigma^2 \times \Sigma^2 \times \Sigma^2$ to denote the inputs that satisfy that promise.

A key idea we pursue is to interpret $I$ and $P$, not just as subsets of inputs, but also as order-3 (3-dimensional) $\{0,1\}$-valued tensors. For example, letting $N = |\Sigma|$, we can represent $\mmP$ as a $\{0,1\}$-valued tensor $\mmP \in \F^{N^2 \times N^2 \times N^2}$, with the tensor entry being one whenever the corresponding tuple is a legal input.

As it turns out, $\mmP$ is \emph{exactly} the matrix multiplication tensor, the same tensor which is extensively studied in the theory of fast matrix multiplication algorithms! We leverage this extremely surprising, if only formal connection, to transfer ideas which have been developed in studying fast matrix multiplication to give new insights into NOF communication complexity, and conversely, to use results proved originally about NOF communication complexity and obtain new alternative proofs of consequential properties of the matrix multiplication tensor.

We use this connection to prove the following results.

    \paragraph{Communication Protocols from the Laser Method} We use the Laser method, the technique introduced by Strassen and used to design the best-known matrix multiplication algorithms~\cite{S86,CW87}, as a tool to design NOF communication protocols. For any three NOF problems $I, J, K$, let $I \otimes J \otimes K$ (``their product'') denote the NOF problem of solving the three simultaneously, i.e., given one input to each, determining whether all three are accepting inputs. Using the Laser method, we show that for any three permutation NOF problems $I, J, K$ with $n$-bit inputs, the product $I \otimes J \otimes K$ (which has $3n$-bit inputs) can be solved using only $n$ bits of communication deterministically.
    
    For example, for any permutation NOF problem $I$ with $n$-bit inputs, one can consider the problem of, given $N$ simultaneous inputs to $I$, determining if they are all accepting inputs. A consequence of our protocol is that this can be solved with $\frac13 n N + \Oh(n)$ bits of communication deterministically. This implies constructions of surprisingly large corner-free sets for powers of arbitrary abelian groups $G$.
    
    By contrast, we show using a counting argument that a random NOF permutation problem $I$ with $n$-bit inputs has requires deterministic communication complexity $\geq n/3 - \Oh(1)$. It was not previously clear to what extent this is tight; our protocol shows that this bound can at least be achieved for any permutation problem which is itself the product of three permutation problems.
    
    Our protocol should also be contrasted with the proof by probabilistic method in \cite{CFTZ21}, which showed that every NOF permutation problem with $n$-bit inputs has a protocol with communication complexity $n/2 + \Oh(1)$, improving in turn upon the trivial $n$ upper bound.
    %\jnote{We didn't define a permutation problem yet, nor product, nor asymptotic communication complexity...}
    %\jnote{What about corner-free sets for non-abelian groups? Is this a thing? Are they equivalent to protocols?}
    
     %In contrast, we show that random NOF permutation problems over the alphabet $[2^n] \times [2^n] \times [2^n]$ requires at least $n/3 - \Oh(1)$ communication.
    
    \paragraph{The Zeroing Out Subrank of Matrix Multiplication}
    A quantity which appears frequently in the theory of matrix multiplication algorithms is the \emph{zeroing out subrank} of matrix multiplication, denoted $Q_{zo}(\langle n,n,n \rangle)$. This is the size of the largest identity tensor (3-dimensional analogue of the identity matrix) which can be achieved as a zeroing out of the tensor $\langle n,n,n \rangle$ for multiplying two $n \times n$ matrices. Strassen~\cite{S86} famously proved that $Q_{zo}(\langle n,n,n \rangle) \geq n^{2 - o(1)}$, which is roughly as large as possible (since $\langle n,n,n \rangle$ is an $n^2 \times n^2 \times n^2$ tensor), and used it as part of the Laser method to design a faster algorithm for matrix multiplication\footnote{See, e.g.,~\cite[Section 8]{B13}. Bounds used in algorithms are sometimes described as bounds on a `monomial degeneration subrank' or `combinatorial degeneration subrank', but these are in turn ultimately used to bound the zeroing out subrank via~\cite{B80}. Note that bounds on more general notions like `border subrank' or `subrank' would not suffice in these algorithms, since the subrank decomposition is applied only to the `outer structure' of a tensor in the Laser method~\cite{S86}.}. Recent \emph{barrier} results have also used this bound to rule out certain approaches to designing faster matrix multiplication algorithms~\cite{AW21,A21,CVZ21}.

    In the bound $Q_{zo}(\langle n,n,n \rangle) \geq n^{2 - o(1)}$, it is natural to ask what is hidden by the $o(1)$. For instance, it was recently shown that the `border subrank' of $\langle n,n,n \rangle$ is $\lceil \frac34 n^2 \rceil$~\cite{KMZ20}; is it possible that we also have $Q_{zo}(\langle n,n,n \rangle) \geq \Omega(n^2)$? This could improve low-order terms in some applications.
    
    In fact, we prove that this is not possible, and that $Q_{zo}(\langle n,n,n \rangle) < n^2 / \omega(1)$. To prove this, we build off of the work of~\cite{LPS17}, who showed a NOF communication lower bound, that any NOF permutation problem requires $\omega(1)$ deterministic communication. Using our connection between NOF communication and matrix multiplication, we are able to modify their proof to get our upper bound on $Q_{zo}(\langle n,n,n \rangle)$.
    
    Interestingly, the prior work~\cite{LPS17} proved their communication lower bound using a connection with the Ruzsa–Szemerédi problem from extremal combinatorics. Not only does our new proof show a connection between $Q_{zo}(\langle n,n,n \rangle)$ and the Ruzsa–Szemerédi problem, but we are in fact able to show that the two are \emph{equivalent}! More precisely, it is impossible to determine the precise asymptotics of $Q_{zo}(\langle n,n,n \rangle)$ without also resolving the Ruzsa–Szemerédi problem. This is interesting in contrast with the `border subrank' which has been precisely determined~\cite{KMZ20}.
    
     Finally, using the connection with NOF communication, we also provide an alternative proof of Strassen's lower bound $Q_{zo}(\langle n,n,n \rangle) \geq n^{2 - o(1)}$, deducing it as a corollary of the existence of the efficient deterministic NOF protocol for $\eval_{\bZ_N}$.

     \paragraph{Slice-Rank Methods}

     Since their introduction~\cite{CFL83}, NOF problems $I$ have been associated with a corresponding 3-hypergraph $H$ whose chromatic number exactly characterizes the deterministic communication complexity of $I$. One recent powerful approach for bounding the chromatic number of such a hypergraph is the `slice rank' method, which was introduced to resolve the cap-set conjecture~\cite{CLP17,ED17}. Roughly speaking, if one can show that the slice rank of the adjacency tensor of $H$ is not too large, then this implies the chromatic number of $H$ is not too small, and hence a communication lower bound.
     
     Since then, a few successful techniques have been introduced to bound the slice ranks of many different tensors~\cite{T16,BCHGNSU17}. Recent work~\cite{CFTZ22} asked whether these techniques could give linear deterministic NOF communication lower bounds. However, they showed that, unfortunately, it is impossible to use such techniques to prove communication lower bounds for any $\textsc{Eval}$ group problem. They proved this by defining the corresponding adjacency tensor and calculating that it has maximal `asymptotic subrank' and hence large slice rank.

     We generalize their result and show that slice-rank methods cannot be used to prove a linear deterministic communication lower bound for \emph{any} NOF permutation problem (not just an $\textsc{Eval}$ group problem). To show this, we observe that for any NOF permutation problem, the corresponding adjacency tensor is actually (a permutation of) the matrix multiplication tensor $\mmP$. Since the matrix multiplication tensor is known to have maximal asymptotic subrank~\cite{S86}, it follows that the adjacency tensor never has small enough slice rank. In particular, this implies that in hindsight, the tensor whose asymptotic subrank was computed in the prior work~\cite{CFTZ22} is exactly the matrix multiplication tensor.

     By contrast, we use slice-rank methods to prove linear lower bounds for other promise problems. The key is that this method works whenever the \emph{promise tensor} has low slice-rank. For example, consider the promise problem $\eq_{\F_3^N}$, where parties receive inputs $x,y,z \in \F_3^N$ with a promise that $x+y+z = 0$, and they wish to decide whether $x=y$. We prove that this problem requires $\Omega(N)$ communication. This may seem like it should be straightforward to prove (``how could the third player's input $z=-(x+y)$ possibly help the first two players to test for equality?''), but we prove that the problem is actually \emph{equivalent} to the cap-set problem (which was a long-open conjecture).

     \paragraph{General Rank Methods and the Asymptotic Spectrum of Tensors}

     Slice-rank is one way to measure the `complexity' of a tensor, but there are many others, including other notions of tensor rank, and more generally a spectrum of different measures which is well-studied in algebraic complexity and tensor combinatorics called the `asymptotic spectrum of tensors'~\cite{S86,WZ22}. Let $r$ be any of these measures (the reader may want to focus on the case when $r$ is the tensor rank). We generalize the slice-rank method and prove that for any promise problem $(I,P)$, its deterministic communication complexity is at least $\log_2(r(I) / r(P))$. Hence, finding a measure $r$ which is larger for the problem $I$ than the promise $P$ leads to a communication lower bound. (This is why, for slice rank, the fact that the promise $\mmP$ has maximal slice rank means that a lower bound cannot be proved in this way.)

     This already has intriguing consequences when $r$ is tensor rank $R$.
     
     For number-in-hand communication, we know that the promise tensor $P_{NIH}$ is the all-1s tensor which has rank $R(P_{NIH}) = 1$. Hence, for any number-in-hand problem $I$, we get a communication lower bound of $\log_2(R(I))$. This is analogous to the standard upper bound in two-party deterministic communication compleixty by the log of the rank of the communication matrix. Moreover, any $2^n \times 2^n \times 2^n$ tensor $I$ (corresponding to an $n$-bit input problem) which is not `degenerate' in some way has rank at least $2^n$, which explains why nearly all number-in-hand communication problems require linear communication. 

     For NOF communication for $n$-bit inputs, as we've discussed, the promise tensor $P_{NOF}$ is the matrix multiplication tensor. Determining the rank of matrix multiplication is the \emph{central question} when designing matrix multiplication algorithms; the matrix multiplication exponent $\omega$ is defined exactly such that the rank of $P_{NOF}$ is $R(P_{NOF}) = (2^n)^{\omega + o(1)}$. Hence, for any problem $I$ with rank $R(I) \geq (2^n)^c$, we get a communication lower bound of $(c - \omega - o(1)) \cdot n$. The fact that $\omega$ appears negated in this lower bound is intriguing: this means that designing faster matrix multiplication algorithms, and hence lowering the known upper bound on $\omega$, leads to an improved communication lower bound for such problems $I$. 

     \paragraph{Matrix Multiplication as a Communication Problem}

     Finally, we consider how other promise problems can be used to shed further light on NOF communication lower bounds. Let $I$ be any $n$-bit NOF problem, let $N = 2^n$, and let $T_{(\mathbb{Z} / N\mathbb{Z})^2}$ denote the structure tensor of the group $(\mathbb{Z} / N\mathbb{Z})^2$, i.e., the tensor
$$T_{(\mathbb{Z} / N\mathbb{Z})^2} = \sum_{a,b,c,d \in [N]} x_{(a,b)} y_{(c,d)} z_{(a+c\pmod{N}, b+d\pmod{N})}.$$
    We observe that $P_{NOF} \subset T_{(\mathbb{Z} / N\mathbb{Z})^2}$ (after appropriately permuting variables), and so it is well-defined to consider three different promise problems using these three tensors: $(I,P_{NOF})$, $(P_{NOF},T_{(\mathbb{Z} / N\mathbb{Z})^2})$, and $(I,T_{(\mathbb{Z} / N\mathbb{Z})^2})$. Moreover, a simple triangle inequality-style argument shows that their communication complexities are related via:
    $$\cc(I,P_{NOF}) \geq \cc(I,T_{(\mathbb{Z} / N\mathbb{Z})^2}) - \cc(P_{NOF},T_{(\mathbb{Z} / N\mathbb{Z})^2}).$$
    On the left-hand side is exactly the deterministic NOF complexity of the problem $I$ that we would like to prove a lower bound for. On the right-hand side, $(I,T_{(\mathbb{Z} / N\mathbb{Z})^2})$ is a similar problem to the original NOF problem of $I$, but with a weaker promise, and so it should be easier to prove communication lower bounds for. What about $\cc(P_{NOF},T_{(\mathbb{Z} / N\mathbb{Z})^2})$? (Recall that in this problem, we're given as input a term from the tensor $T_{(\mathbb{Z} / N\mathbb{Z})^2}$, and our goal is to determine whether or not it is a valid NOF input from $P_{NOF}$.)

    For this problem, we give a protocol with communication $\leq n$ (the trivial bound would be $2n$). Moreover, we prove a communication lower bound of $\geq (\omega - 2) n$ (this follows from the rank-based approach discussed above). So, further improvements to the communication protocol actually \emph{require} fast matrix multiplication (since they would imply an upper bound on $\omega$ via this inequality). If it were the case that $\omega = 2$ (which is popularly conjectured), then it is plausible that this problem actually has sublinear communication complexity. In that case, in order to prove a linear lower bound for the NOF problem $(I,P_{NOF})$, it would suffice to prove a linear lower bound for the seemingly-harder problem $(I,T_{(\mathbb{Z} / N\mathbb{Z})^2})$. In summary, if we could design a faster matrix multiplication algorithm achieving $\omega=2$, then this may be a promising approach to proving NOF communication lower bounds.

\section{Tensor Preliminaries}

In this paper we will relate communication problems to related tensors and their properties. We will focus on order-3 tensors (also known as 3-dimensional tensors). For a field $\F$ and three positive integers $A,B,C$, we define $\F^{A\times B \times C}$ to be the set of tensors $$T = \sum_{a \in A} \sum_{b \in B} \sum_{c \in C} T_{a,b,c} x_a y_b z_c$$ for coefficients $T_{a,b,c} \in \F$ (where $x_a, y_b, z_c$ are formal variables). When the specific sets are not important, we will often write $\F^{|A|\times |B| \times |C|}$ instead of $\F^{A\times B \times C}$.

$T$ is $\{0,1\}$-valued if, for all $a,b,c$, we have $T_{a,b,c} \in \{0,1\}$. Such tensors are in bijection with subsets of $A \times B \times C$. All tensors considered in this paper are assumed to be $\{0,1\}$-valued unless stated otherwise.

The field $\F$ that one works over will be primarily important to us for the following definition. The (not necessarily $\{0,1\}$-valued) tensor $T$ has rank 1 if it can be written as $$T = \left( \sum_{a \in A} \alpha_a x_a \right)\left( \sum_{b \in B} \beta_b y_b \right)\left( \sum_{c \in C} \gamma_c z_c \right)$$ for coefficients $\alpha_a, \beta_b, \gamma_c \in \F$. More generally, the rank of the (not necessarily $\{0,1\}$-valued) tensor $T$, denoted $R(T)$, is the minimum number of tensors of rank 1 which sum to $T$.

Rank is critical to the definition of the matrix multiplication exponent $\omega$. For $n,m,p \in \mathbb{N}$, we write $\langle n,m,p \rangle \in \F^{n^2 \times n^2 \times n^2}$ to denote the matrix multiplication tensor $$\langle n,m,p \rangle = \sum_{i =1}^n \sum_{j=1}^m \sum_{k=1}^p x_{i,j} y_{j,k} z_{k,i}.$$ The exponent $\omega$ is defined as $$\omega := \inf\{ t \mid R(\langle n,n,n \rangle) \leq n^t \text{ for some } n \in \mathbb{N}\}.$$ Indeed, it is known that for any $\varepsilon>0$, matrix multiplication of $n \times n$ matrices can be performed with an arithmetic circuit of size $\Oh(n^{\omega + \varepsilon})$, and conversely, that any arithmetic circuit for matrix multiplication can be converted into a tensor rank upper bound which yields the same operation count in this way~\cite{S73}.

We say $T$ is an identity tensor if, for all $a \in A$, there is at most one pair $(b,c) \in B \times C$ such that $T_{a,b,c} \neq 0$, and similarly for all $b \in B$ and for all $c \in C$. Often, when $A=B=C$, one specifically thinks of the identity tensor of size $k$ as a tensor $\sum_{a \in S} x_a y_a z_a$ for some subsets $S \subset A$ of size $|S|=k$. To be clear that this is sufficient but not necessary (as one could in general permute the indices), we will refer to permutations of identity tensors.

If $T,T' \in \F^{A \times B \times C}$, we say $T'$ is a zeroing out of $T$ if there are subsets $A' \subset A, B' \subset B, C' \subset C$ such that $T'_{a,b,c}$ is equal to $T_{a,b,c}$ whenever $a \in A', b \in B', c \in C'$, and $T'_{a,b,c} = 0$ otherwise. In other words, we are setting coefficients outside of $A',B',C'$ to zero. We will sometimes write $T' = T|_{A',B',C'}$.

The zeroing out subrank of $T$, denoted $Q_{zo}(T)$, is the maximum $k$ such that $T$ has a zeroing out to a permutation of an identity tensor of size $k$. This more combinatorial variant on the notion of `subrank' (which uses `restrictions' rather than zeroing outs) will be useful a number of times.

\section{Promise problems and colored tensors}

\begin{definition}
A promise problem over alphabet $\Sigma$ is a pair of subsets $I \subset P \subset \Sigma\times\Sigma\times\Sigma$, where $P$ is a set of allowed inputs (the \emph{promise}), and $I$ is the subset of accepting inputs (the \emph{problem}).
\end{definition}

To study those objects, we will introduce a notion of \emph{colored tensors}. A colored tensor is a pair $(I, T_P)$, where $T_P \in \F^{\Sigma \times \Sigma \times \Sigma}$ is a tensor and $I \subset \supp(T_P)$ is a subset of its non-zero terms. (We imagine those terms to have a special color, and in what follows we will call them \emph{green} terms). 3-party promise number in hand problems over the alphabet $\Sigma$ are therefore in direct correspondence with $\{0, 1\}$-valued colored tensors over $\F^{\Sigma \times \Sigma \times \Sigma}$, where the non-zero terms of the tensor $T$ corresponds to the set of allowed inputs, and the subset $I$ of green terms corresponds to accepting inputs.

%\jnote{We need a good name}
\begin{definition}
For any field $\F$ we can define the \emph{communication tensor} of a promise problem $(I, P)$ as a colored tensor $(I, T_P)$, with $T_P \in \F^{\Sigma \times \Sigma \times \Sigma}$ given by
\begin{equation*}
    T_P := \sum_{(i,j,k) \in P} x_i y_j z_k.
\end{equation*}

We will say that $T_P$ is the \emph{promise tensor} of the problem $(I, P)$. 
\end{definition}

We will sometimes abuse notation, and use the same symbol to denote a $\{0, 1\}$-valued tensor $P \in \F^{A \times B \times C}$ and its support $\supp(P) \subset A \times B \times C$ whenever it is clear from context. All tensors we discuss in this paper are $\{0,1\}$-valued and for the vast majority of the discussion the underlying field is irrelevant.

\paragraph{Communication Complexity Model}
We use $\cc(I, P)$ to denote the deterministic communication complexity of the promise problem $(I, P)$. To make this concrete, we consider a shared blackboard model: here a protocol is given by a triple of functions, one for each player, encoding what each player should append to the blackboard in each round of communication, given the current state of the blackboard, and their private input. Players write their communication on the blackboard in cyclic order; in each stage each player can either append something to the blackboard, accept the current instance, or reject the current instance. A specific instance is accepted if \emph{all players} accept it (based on their own input, and the final state of the blackboard). A protocol solves a promise problem $(I, P)$ if it accepts all the inputs from $I$ and rejects all other inputs from $P \setminus I$. (Since $P$ is the promise, the protocol need not have any particular behavior on inputs outside of $P$.)

\paragraph{Asymptotic communication complexity}
For a promise problem $(I,P)$ we use $(I, P)^{\otimes n}$ to denote a promise problem in which parties are given $n$ valid instances of the original promise problem $(I, P)$, and wish to determine whether \emph{all} of those instances simultaneously are accepting.

More formally, $(I, P)^{\otimes n}$ is the promise problem with the promise $P^{\otimes n} \subset (\Sigma^n) \times (\Sigma^n) \times (\Sigma^n)$, i.e., for $A,B,C \in \Sigma^n$, the triple of sequences $(A, B, C)$ is in $P^{\otimes n}$ if and only if we have $(A_i, B_i, C_i) \in P$ for all $i \in [n]$. The set of accepting instances $I^{\otimes n}$ is defined analogously. This operation corresponds to the Kronecker product of the associated communication tensors.
\begin{definition}
    For tensors $I^{(1)} \subset T^{(1)} \subset \F^{A_1 \times B_1 \times C_1}$ and $I^{(2)} \subset T^{(2)} \subset \F^{A_2 \times B_2 \times C_2}$ We define the Kronecker product of colored tensors $(I^{(1)}, T^{(1)}) \otimes (I^{(2)}, T^{(2)})$ as a colored tensor $(I, T)$ in $\F^{A_1 A_2 \times B_1 B_2 \times C_1 C_1}$, where $T = T_1 \otimes T_2$ satisfies 
    \begin{equation*}
    T_{(a_1,a_2), (b_1, b_2), (c_1, c_2)} = T^{(1)}_{a_1,b_1,c_1} \cdot T^{(2)}_{a_2,b_2,c_2},
    \end{equation*}
    and $I = \{(a_1, a_2), (b_1, b_2), (c_1, c_2) : (a_i, b_i, c_i) \in I_i \text{ for } i \in \{1,2\}\}$.
\end{definition}

We remark that all tensors considered in this work are order-3 tensors (also known as 3-dimensional tensors), and we will use symbol $\otimes$ exclusively to denote the Kronecker product of tensors (as opposed to tensor product).

We will define the asymptotic communication complexity of the problem $P$ as
\begin{equation*}
    \acc(I,P) := \lim_{n\to\infty} \frac{\cc( (I, P)^{\otimes n})}{n}.
\end{equation*}

\paragraph{Permutation problems}
We pay special attention to a particular class of communication problems, which we will call \emph{injection problems} --- as we will see later on, the deterministic communication complexity of injection problems have very simple combinatorial interpretation, and they all have protocols with constant randomized communication (under public randomness).

\begin{definition}[Injection problem]
    Let $\pi_1, \pi_2, \pi_3 : \Sigma^3 \to \Sigma$ be projections onto first, second, and third coordinate respectively.
    
    We say that $(I, P)$ (where $I \subset P \subset \Sigma^3$) is an \emph{injection problem},
    if the tensor $T_I := \sum_{(i,j,k) \in I} x_i y_j z_k$  restricted to the coordinates appearing in $I$ (i.e. $(T_I)|_{\pi_1(I), \pi_2(I), \pi_3(I)}$) is a permutation of an identity tensor.
    
    That is: for every valid input of the $i \in \Sigma$ of the first player, there is at most one pair $j, k$, such that $(i,j,k) \in I$, and analogously for the second and the third player.
\end{definition}

\begin{definition}[Permutation problem]
    We say that $(I, P)$ is a permutation problem if $(I, P)$ is an injection problem and moreover $|I| = \Sigma$.
\end{definition}

Note that if $(I, P)$ is an injection (resp. permutation) problem, then $(I, P)^{\otimes n}$ also is an injection (permutation) problem.

\subsection{Number on the forehead problems as promise-number-in-hand problems}
We can treat a NOF problem $I \subset \Sigma \times \Sigma \times \Sigma$ as a promise problem over the alphabet $\Gamma := \Sigma \times \Sigma$, where the promise $\mmP \subset \Gamma^3$ is given by the following promise-tensor
\begin{equation*}
    \mmP := \sum_{i,j,k \in \Sigma} x_{i,j} y_{j,k} z_{k,i}.
\end{equation*}
A key observation that we will explore in this paper is that $\mmP$ is exactly the matrix multiplication tensor $P = \inprod{N,N,N}$ where $N = |\Sigma|$.

The set of accepting instances $I \subset \Gamma \times \Gamma \times \Gamma$ can be bijectively mapped into a subset $\tilde{I} \subset \inprod{N,N,N}$ of non-zero terms of the promise tensor, by $\tilde{I} := \{ ((x,y), (y,z), (z,x) ) : (x,y,z) \in I \}$.

Having interpreted number-on-the-forehead problems as number-in-hand problems with promise $\inprod{N,N,N}$, we observe that our more general definition of the permutation problem agrees with the definition of the permutation problem in \cite{LPS17} for NOF problems --- which are in turn are equivalent to Latin squares of size $N \times N$.

An important family of examples of NOF permutation problems are $\eval$ problems induced by an abelian group $G$. Specifically, for any abelian group $G$, let $\eval_G$ be the problem defined as follows. Players have $x,y,z \in G$ on their foreheads, and they want to decide whether $x+y+z = 0$. We can represent the problem by its set of accepting instances $I \subset G^3$ defined by $I = \{(x,y,z) : x+y+z = 0\}$. %, or equivalently we can treat it as a promise problem represented by colored tensor $(I', \inprod{n,n,n})$ where $n = |G|$, and $I' \subset \supp{\inprod{n,n,n}}$ is given by $I' := \{ ((x,y), (y,z), (z,x)) \in (G\times G)^3 : x+y+z = 0 \}$.

The observation that all NOF permutation problems admit an efficient randomized protocol, generalizes to injection problems regardless of the promise, using a constant communication randomized protocol for two-players $\textsc{Equality}$ problem.
\begin{fact}
Every injection problem has a $O(1)$ randomized protocol with shared randomness, regardless of the promise.
\end{fact}
\begin{proof}
First player, upon receiving his input $x$ can find unique $y_1, z_1$, such that $(x,y_1, z_1) \in I$. He can then proceed by using the randomized two-player protocol for $\textsc{Equality}$ to check whether $y_1 = y$, and $z_1 = z$, where $y,z$ are inputs of the second and third player respectively.
\end{proof}

What can we say about the deterministic communication complexity of promise injection problems? It turns out that we can express it equivalently as a question about the smallest proper coloring of the communication tensor $(I, P)$.
\begin{definition}
Consider a colored tensor $T = (I, P)$, We say that a subset of its green terms $S \subset I$ is an independent set if $S$ is a permutation of an identity tensor and the restriction $T|_{\pi_1(S), \pi_2(S), \pi_3(S)}$ is equal to $S$. We use $\alpha(T)$ to denote the size of largest independent set in the colored tensor $T$.
\end{definition}

If $T$ is an ordinary tensor, we use $\alpha(T)$ to denote the size of the largest independent set in the colored tensor $(\supp(T), T)$ (i.e., we treat all its terms as green). This coincides with the usual notion of zero-out subrank of a tensor.

Note that if $T = (I, P)$ is an injection problem, $S \subset I$ is an independent set if and only if all non-zero terms in the restriction $T|_{\pi_1(S), \pi_2(S), \pi_3(S)}$ are in $I$ --- that is, if we denote by $S_i := \pi_i(S)$ then set $S$ is independent if and only if 
\begin{equation*}
    \supp(T) \cap (\pi_1(S_1)^{-1} \times \pi_2(S_2)^{-1} \times \pi_3(S_3)^{-1}) = S.
\end{equation*}

With a definition of an independent set, we can define a proper coloring and the chromatic number of a colored tensor in a natural way.
\begin{definition}
For a colored tensor $T := (I, P)$  we say that $\tau : I \to [k]$ is a proper coloring if for all colors $c\in [k]$, the set $\tau^{-1}(c)$ is independent.

We denote by $\chi(T)$ the \emph{chromatic number} of tensor $T$ --- the smallest $k$, such that there is a proper coloring $\tau : I \to [k]$. Again, if $T$ is an ordinary tensor, we use a shorthand $\chi(T) := \chi(\supp(T), T)$.
\end{definition}

The proof of the following characterization follows exactly the original proof in \cite{CFL83}, just phrased in a slightly more general language of promise problems instead of specifically Number on the Forehead problems.

\begin{theorem}
\label{thm:coloring-is-communication}
The deterministic communication complexity of any promise injection problem $(I, P)$ is equal to $\lceil \log \chi(I,P) \rceil$.
\end{theorem}
\begin{proof}
Given a coloring $\tau : I \to [k]$, we will describe a communication protocol with complexity $\lceil \log k \rceil$. The first player, with an input $x$, looks at the only $y^{(1)},z^{(1)}$, s.t. $(x,y^{(1)},z^{(1)}) \in I$, and writes down the color $\tau(x,y^{(1)},z^{(1)})$. Similarly, the second player, with an input $y$, compares the written color with the color of $(x^{(2)}, y, z^{(2)}) \in I$, and analogously the third player. Clearly, if $(x,y,z)$ was an accepting instance, all three of the triples in the consideration would be the same ($x^{(2)} = x^{(3)} = x, y^{(1)} = y^{(3)} = y, z^{(1)} = z^{(2)} = z$), and therefore the colors under consideration would agree --- the protocol will correctly accept such an instance.

On the other hand, if all $(x, y^{(1)}, z^{(1)}), (x^{(2)}, y, z^{(2)}), (x^{(3)}, y^{(3)}, z)$ have the same color $c$, let us take the independent set $S := \tau^{-1}(c) \subset I$, then $x \in \pi_1(S), y \in \pi_2(S), z \in \pi_3(S)$, and since $(x,y,z) \in P$ and $S$ is an independent set, we  must have $(x,y,z) \in I$ --- i.e. the instance was accepting. This completes the proof that a coloring $\tau : I \to [k]$ implies existence of an efficient protocol.

The converse is similar: given a protocol for a problem $(I, P)$ with communication at most $k$, we can use the transcript of the communication on each input $(x,y,z) \in I$ as a color for the term. This yields a mapping $\tau : I \to [2^k]$, and all we need to check is that it is a proper coloring. Indeed, let us chose an arbitrary color $c$, and take a set $S = \tau^{-1}(c)$. If there was $(x, y, z) \in P - I$, such that $(x, y, z) \in \pi_1(S) \times \pi_2(S) \times \pi_3(S)$, the communication transcript on the (allowed) input $(x,y,z)$ would have been the same as on accepting inputs $(x, y^{(1)}, z^{(1)}), (x^{(2)}, y, z^{(2)}), (x^{(3)}, y^{(3)}, z)$, and therefore each player would have accepted it --- contradicting correctness of the protocol, since $(x,y,z) \not\in I$.
\end{proof}

\subsection{Comparison with standard bounds for the NOF communication complexity}
What we discussed so far is a slightly different perspective on the standard way of characterizing NOF communication complexity (already present in \cite{CFL83}). Classically, with  NOF permutation problem $I \subset \Gamma \times \Gamma \times \Gamma$, one associated a directed 3-hypergraph $H$, and the logarithm of the chromatic number of this hypergraph is equal to the deterministic communication complexity of the associated problem.

We will show that in fact the adjacency tensor of the aforementioned hypergraph $H$ is just a permutation of the promise tensor $\mmP$ (i.e., the matrix multiplication tensor), and the permutation is uniquely determined by the problem $I$ --- it corresponds to permuting variables $y_i$ and $z_i$ in a way such that $I$ is exactly on the diagonal of the tensor.

%\jnote{Before saying that this is a reason why slice-rank approach fail, maybe it is worth mentioning what is the slice-rank approach...}
This provides a generalization, and another perspective on the barrier against using a slice-rank or any other bounds on the subrank of the adjacency tensor of $H$ as a way to show a lower bound for communication complexity of $H$. 

Specifically, note that any independent set $S \subset V(H)$ induces a zeroing-out of the adjacency tensor of the hypergraph $H$ to a large identity subtensor (given by restriction to indices in $S \times S \times S$ --- note that in the construction of the hypergraph $H$ and its adjacency tensor, we include self-loops $(t,t,t)$ for each $t \in V(H)$). One could then hope to show a lower bound on the chromatic number of the hypergraph $H$ by upper bounding the size of the largest independent set, which in turn can be upper bounded by subrank of adjacency sub-tensor. Finally, it is known that subrank of any tensor is always bounded by its slice-rank, and a few successful techniques for controlling slice-ranks of particular tensors of interest have been developed~\cite{T16,BCHGNSU17}.

This, apparently promising avenue for showing lower bounds for specific NOF permutation problems unfortunately has to fail. In \cite{CFTZ21} it was shown that similar techniques as those used by Strassen to prove lower bound on the asymptotic subrank of the matrix multiplication tensor, can be used to lower bound the asymptotic subrank of the adjacency tensor for any $\textsc{Eval}$ group problem --- it is maximally large. We observe that this is not a coincidence: in fact all adjacency tensors constructed this way are just permutations of the matrix-multiplication tensor, so the lower bound for asymptotic subrank of the matrix multiplication tensor directly applies --- not only for problems arising as $\eval_G$, but more generally for all number-on-the-forehead permutation problems.

\begin{definition}[Communication Hypergraph \cite{CFL83}]
\label{def:communication-hypergraph}
For a permutation NOF problem $I \subset \Gamma \times \Gamma \times \Gamma$, we define its communication hypergraph $H(I)$ to be an order-3 hypergraph, with vertices $V(H(I)) = I$ --- set of all accepting instances, and hyperedges constructed in the following way: for all $(x,y,z) \in \Gamma \times \Gamma \times \Gamma$, let $T_1 := (x, y, z') \in I$ be the only triple in $I$ with agreeing with $(x,y,z)$ in the first two coordinates, and similarly $T_2 := (x',y,z) \in I, T_3 := (x, y', z) \in I$.

Then $E(H(I)) = \{ (T_1, T_2, T_3) : (x,y,z) \in \Gamma^3 \}$.
\end{definition}

\begin{fact}
The adjacency tensor of the hypergraph $H(I)$ is a permutation of the matrix multiplication tensor $\inprod{|\Gamma|, |\Gamma|, |\Gamma|}$. 

Moreover if we consider a set of green terms $\tilde{I} \subset \inprod{|\Gamma|, |\Gamma|, |\Gamma|}$ given by $\tilde{I} = \{((x,y), (y,z), (z,x)) : (x,y,z) \in I\}$ in the colored communication tensor for the problem $I$, this set 
 corresponds under the aforementioned permutation exactly to the set of diagonal terms $\{(t,t,t) : t \in I\}$ of the adjacency tensor of $H(I)$.
\end{fact}
\begin{proof}
Consider a map $\pi_{12} : I \to \Gamma \times \Gamma$, given by $\pi_{12}(x,y,z) = (x,y)$, and similarly $\pi_{23}, \pi_{13}$. Note that those three bijections between $I$ and $\Gamma \times \Gamma$. Applying those bijections to three axes of the adjacency tensor $A$ of the hypergraph $H(I)$, we see that an image of any hyperedge $(T_1, T_2, T_3)$ is in the \Cref{def:communication-hypergraph} is $(\pi_{12}(T_1), \pi_{23}(T_2), \pi_{13}(T_3)) = ( (x,y), (y,z), (x,z) )$ --- a non-zero term of the matrix multiplication tensor $\inprod{|\Gamma|, |\Gamma|, |\Gamma|}$. All the non-zero terms of the matrix multiplication tensor can be obtained this way: since, by construction of the communication hypergraph, any term $( (x,y), (y,z), (x,z))$ can be lifted to an hyperedge $(T_1, T_2, T_3) \in H(I).$

Finally, the statement about the image of the diagonal follows directly from the definition of permutations $\pi_{12}, \pi_{23}$ and $\pi_{13}$ and construction of the hypergraph $H$.
\end{proof}
\begin{corollary}
    For any NOF permutation problem, the adjacency tensor of its communication hypergraph has maximal asymptotic subrank.
\end{corollary}

\subsection{Lower bounding the communication complexity of random NOF permutation problem}

In this section we will show that a random NOF permutation problem over the alphabet $[N] \times [N] \times [N]$, where $N = 2^n$ requires $\geq n/3 - \Oh(1)$ deterministic communication. This is a simple counting argument, similar in spirit to the one used in \cite{DP08}. Since the NOF permutation problems over $[N] \times [N] \times [N]$ are bijective with Latin squares of size $N \times N$, known upper and lower bounds for the number of Latin squares can be used in our proof.

The following is a consequence of the fact that any Latin square can be constructed by repeatedly taking a perfect matching in a regular bipartite graph, and using standard bounds on the number of perfect matchings in a $k$-regular bipartite graph. 
% \jnote{I am not sure if the preceding sentence is necessary. I feel like it's cool fact to know, and can be communicated quite succinctly, on the other hand, isn't it a bit too cryptic? And having a detailed explanation would be a bit too orthogonal to the discussion.}
% \joshnote{Did you find it on mathoverflow or something? I think it would be fine to say something vague (`This follows from a connection with the number of perfect matchings in a $k$-regular bipartite graph') then just link to something with more details.}
\begin{theorem}[\cite{VW01}]
\label{thm:nof-number}
Total number $\#P$ of NOF permutation problems over $[N] \times [N] \times [N]$ satisfies
\begin{equation*}
    N^{N^2} e^{-N^2} \leq \#P \leq N^{N^2}.
\end{equation*}
\end{theorem}

With this in hand, we can prove
\begin{theorem}
For most of NOF permutation problems $(I, P)$, over the alphabet $[N]$ where $N = 2^n$, we have $\cc(I,P) \geq \frac{n}{3} - \Oh(1)$.
\end{theorem}
\begin{proof}
Let $(I, P)$ be a NOF permutation problem over $[N] \times [N] \times [N]$ with communication complexity $d$. By \Cref{thm:coloring-is-communication} communication complexity is equivalent to the existence of proper coloring $\kappa : I \to \{0, 1\}^d$. Since $I$ is a permutation problem, for any $x \in [N] \times [N]$, there is a unique $\pi_1^{-1}(x) \in I \subset ([N] \times [N])^3$ (where $\pi_1 : ([N] \times [N])^3 \to [N] \times [N]$ is the projection onto first coordinate). We can define $\kappa_1 : [N] \times [N] \to \{0, 1\}^d$ as $\kappa_1(x) := \kappa(\pi_1^{-1}(x))$. Analogously, we define $\kappa_2$ and $\kappa_3$.

We observe that the triple of functions $(\kappa_1, \kappa_2, \kappa_3)$ uniquely determines $I$. Indeed, for $(a,b,c) \in P$ to check if $(a,b,c) \in I \iff \kappa_1(a) = \kappa_2(b) = \kappa_3(c)$. Therefore, the number of bits to specify a problem with communication complexity $d$ (i.e. the logarithm of the number of such problems) is at most $3 d N^2$ --- to specify a function $\kappa : [N] \times [N] \to \{0,1\}^d$ we need $d N^2$ bits.

On the other hand, using the lower bound of \Cref{thm:nof-number}, we need $N^2 \log N - \Oh(N^2)$ bits to write down arbitrary NOF permutation problem. Therefore if $d \leq (\log N) / 3 - \Oh(1)$, most of the problems do not have a protocol with communication complexity $d$.
\end{proof}

\subsection{Background on independence number and coloring of matrix multiplication tensor}
The following statement asserting that matrix-multiplication tensor has almost maximal zero-out subrank is an important insight in the study of matrix multiplication. It has been originally proven by Strassen in \cite{S87}, and has been used as a technical step in the design of some of the fast matrix multiplication algorithms. Here we will use it as a tool in the proof of \Cref{thm:laser-method-ub}. On the other hand, in \Cref{sec:mat-subrank} we will discuss how, using our connection between number on the foreheads protocol and structural properties of the matrix multiplication tensor, the following statement can be easily deduced from the known results in the NOF communication complexity on efficient protocols for some permutation problems. 
\begin{theorem}[\cite{S87}]
There is $I \subset \supp{\inprod{N,N,N}}$, such that $|I| \geq N^{2 - o(1)}$, and the restriction $\inprod{N,N,N}|_{\pi_1(I), \pi_2(I), \pi_3(I)}$ is a permutation of the identity tensor, i.e. $Q_{zo}(\inprod{N,N,N}) \geq N^{2 - o(1)}$.

Equivalently, in the notation introduced in this work, 
\begin{equation*}
    \alpha(\inprod{N,N,N}) \geq N^{2 - o(1)}.
\end{equation*}
\end{theorem}

The following is a standard technique of lifting the large independent set to a small coloring in symmetric combinatorial objects. We include the proof for completeness.
\begin{definition}[Automorphism group]
    For a colored tensor $T = (I, P)$, we define $\Aut(T)$ to be the set of all triples of permutations $(\sigma_1, \sigma_2, \sigma_3)$ such that $\sigma(I) = I$, and moreover for all $(x,y,z)$ we have $T_{x,y,z} = T_{\sigma_1(x), \sigma_2(y), \sigma_3(z)}$.
\end{definition}

\begin{lemma}
\label{lem:coloring-from-symmetry}
For a colored tensor $T = (I, P)$ if for any two green elements $a,b \in I$, there is a $\sigma \in \Aut(T)$, s.t. $\sigma(a) = b$, then 
\begin{equation*}
    \chi(T) \leq \frac{|I|}{\alpha(T)} \log |I|.
\end{equation*}
\end{lemma}
%\jnote{Should the following proof even be here? Or should we get rid of it/move to appendix, or something? This is painfully standard...}
%\joshnote{To be honest, I think most of our proofs are not that complicated :)  I think it would be fine to leave it or move it to an appendix, whichever you prefer.}
\begin{proof}
Let $S \subset I$ be an independent set, and $\sigma \sim \Aut(T)$ a uniformly random element of the automorphism group. Then clearly $\sigma(S)$ is also an independent set. We claim that for any $x \in I$, we have $\Pr(x \in \sigma(S)) = \frac{|S|}{|I|}$. Indeed, let $S = \{ s_1, s_2, \ldots s_m \}$. The subset $T \subset \Aut(T)$ of permutations $\sigma$ such that $\sigma(s_1) = x$ is exactly a coset of a subgroup $H$ in $\Aut(T)$ where $H$ is a set of permutations fixing $x$, i.e. $H = \{ \sigma : \sigma \in \Aut(T), \sigma(x) = x \}$. Number of such cosets is exactly $|I|$, since we assumed that $x$ can be mapped to any other $s \in I$ by some permutation in $\Aut(T)$ (and no other element, by the property of automorphism group that each $\sigma \in \Aut(T)$ preserves green terms: $\sigma(I) = I$). Since all cosets of a given subgroup are disjoint and have the same size, we get $\Pr(\sigma(s_1) = x) = 1/|I|)$. Therefore
\begin{equation*}
    \Pr(x \in \sigma(S)) = \sum_{i \leq |S|} \Pr(x = \sigma(s_i)) = |S|/|I|.
\end{equation*}

Finally, if we consider $K = C (|I| \log |S|) / |S|$ independent random permutations $\sigma_1, \sigma_2, \ldots, \sigma_m \in \Aut(T)$, we wish to say that with positive probability for every element $x \in I$ it is covered by at least one set $\sigma_i(S)$. Indeed, in expectations, each of those is covered by $C \log |S|$ such sets, and applying a standard Chernoff and union bound argument, we reach the conclusion. 

Having a covering of $I$ by $K$ independents sets, we can easily find a coloring $\tau : I \to [K]$ mapping an elemnet $x \in I$ to the smallest $i \in [K]$ such that $x \in \sigma_i(S)$.
\end{proof}

\begin{corollary}
\label{cor:mat-mult-coloring}
The chromatic number of the matrix multiplication tensor satisfies
\begin{equation*}
    \chi(\inprod{N,N,N}) \leq N^{1 + o(1)}.
\end{equation*}
\end{corollary}

We shall also need the following standard fact on the Kronecker products of matrix multiplication tensors.
\begin{fact}[see, e.g., {\cite[{Page 24}]{B13}}]
    \label{fct:mat-mult-product}
    The Krnoecker product of matrix multiplication tensors satisfies
    \begin{equation*}
        \inprod{N_1, M_1, P_1} \otimes \inprod{N_2, M_2, P_2} = \inprod{N_1 N_2, M_1 M_2, P_1 P_1}
    \end{equation*}
\end{fact}

\subsection{Laser method gives a non-trivial asymptotic protocol for all NOF permutation problems}
In this section we will use the celebrated Laser method, introduced to design fast matrix multiplication algorithms~\cite{S86,CW87}, to prove the following surprising upper bound on the asymptotic NOF communication complexity of the third player of arbitrary permutation problem.
\begin{theorem}
\label{thm:laser-method-ub}
For every triple of NOF permutation problems $I, J, K$, each over $\{0,1\}^n \times \{0,1\}^n \times \{0,1\}^n$, we have
\begin{equation*}
    \cc(I \otimes J \otimes K) \leq (1 + o(1))n.
\end{equation*}

(Note that the trivial bound is $\cc(I \otimes J \otimes K) \leq 3n$.)

In particular for any NOF permutation problem $P$ over $\{0,1\}^n \times \{0,1\}^n \times \{0,1\}^n$, we have 
\begin{equation*}
    \acc(P) \leq (1 + o(1))\frac{n}{3}.
\end{equation*}
\end{theorem}

% We can moreover define the "zero-out" partial order: $(I_1, T_1) \leq_{zo} (I_2, T_2)$ where $T_i \in \F^{A_i \times B_i \times C_i}$ if there are injections $\iota_A : A_1 \to A_2, \iota_B : B_1 \to B_2, \iota_C : C_1 \to C_2$, s.t. $\forall i,j,k (T_1)_{i,j,k} = (T_2)_{\iota_A(i), \iota_B(j), \iota_C(k)}$, and moreover $\iota(I_1) \subset I_2$.

We define the \emph{block colored tensor} to be a colored tensor $(I, T)$ over $A \times B \times C$ together with a collection of partitions $A = A_1 \cup \ldots A_p$, $B = B_1 \cup \ldots \cup B_q$, and $C = C_1 \cup \ldots C_q$.

The \emph{outer structure} of a block tensor $T$, is a tensor $\Out(T)$ over $[p] \times [q] \times [r]$ with $
\Out(T)_{i,j,k} = 1$ if and only if the restriction $T|_{A_i, B_j, C_k} \not= 0$ (and $\Out(T)_{i,j,k} = 0$ otherwise). The \emph{inner structure} $\Inn(T)$ of a block tensor $T$ is a collection of (colored) tensors $T|_{A_i, B_j, C_k}$. The notions of outer and inner structure feature prominently in the Laser method; see e.g.,~\cite[Section 8]{B13}.

\begin{lemma}
\label{lem:block-coloring}
For a block-colored tensor $T$, we have
\begin{equation*}
    \chi(T) \leq \chi(\Out(T)) \cdot \max_{T' \in \Inn(T)} \chi(T').
\end{equation*}
\end{lemma}
\begin{proof}
    This is immediate --- to find a proper coloring for the set of green terms $I$ of a block tensor $T$, we can map an element $x \in I$ to $(c_1, c_2)$, where $c_1$ is a color assigned to the block in which $x$ lies in the coloring of $\Out(T)$, and $c_2$ is a color assigned to $x$ in the coloring of the block itself.

    We need to show that each color class is an independent set. With an induced partition of $A, B, C$, this amounts to observing that for any block tensor for which $\Out(T)$ is is a permutation of the identity tensor, as well as all tensors in $\Inn(T)$, the entire tensor also is a permutation of the identity tensor. 
\end{proof}

\begin{fact}
\label{fct:block-product}
Let $T^1, T^2$ be a pair of block-colored tensor. Then $T := T^1 \otimes T^2$ is a block colored tensor, with the outer structure $\Out(T) = \Out(T^1) \otimes \Out(T^2)$, and the inner structure $\Inn(T) \subset \Inn(T^1) \otimes \Inn(T^2)$ (Here, for two collections of tensors $\mathcal{A}$ and $\mathcal{B}$ we use $\mathcal{A} \otimes \mathcal{B} := \{ a \otimes b : a \in \mathcal{A}, b\in \mathcal{B}\}$.)
\end{fact}
\begin{proof}
    Directly follows from the definition.
%    \jnote{Should we prove something here?}
\end{proof}

\begin{fact}
\label{fct:chromatic-numbers-product}
For any pair of colored tensors $T_1, T_2$, we have $\chi(T_1 \otimes T_2) \leq \chi(T_1) \chi(T_2)$.
\end{fact}
\begin{proof}
    Given two colorings $\kappa_1 : I_1 \to [k_1]$ and $\kappa_2 : I_2 \to [k_2]$ we can construct a coloring of the Kronecker product $T_1 \otimes T_2$ using pairs $[k_1] \times [k_2]$, by mapping a term $( (i_1, i_2), (j_1, j_2), (k_1, k_2)) \in I_1 \otimes I_2$ to a pair $(\kappa_1(i_1, j_1, k_1), \kappa_2(i_2,j_2,k_2))$. The fact that it is a proper coloring follows directly from the fact that Kronecker product of two identity tensors is again an identity tensor.
\end{proof}

\begin{lemma}
\label{lem:block-perm}
There is a triple of partitions $\mathcal{A}, \mathcal{B}, \mathcal{C}$ such that for any NOF permutation-problem $(I,\mmP)$, the block-colored tensor $T := (I, \mmP, \mathcal{A}, \mathcal{B}, \mathcal{C})$ induced by those partitions satisfies the following conditions
\begin{itemize}
    \item The outer structure $\Out(T) = \inprod{N, 1, 1}$.
    \item Every colored tensor in the inner structure $(Y, T') \in \Inn(T)$ has $T' = \inprod{1,N,N}$ and $\chi(Y, T') = 1$.
\end{itemize}
\end{lemma} 
\begin{proof}
    Consider partitions $\mathcal{A} := (A_1, \ldots, A_n)$ where $A_t := \{ (i, t) : i \in [N] \}$, similarly $\mathcal{B} := (B_1, \ldots B_N)$ where $B_t = \{ (t, i) : i \in [N] \}$, and finally $\mathcal{C} := (C) := ([N] \times [N])$ is a trivial partition with a single part.

    First note that with this partition of the matrix multiplication tensor $P$, the block $A_{t_1} \times B_{t_2} \times C$ is non-empty if and only if $t_1 = t_2$,
    \begin{equation*}
        \Out(T) = \sum_{t} x_{1,t} y_{t, 1} z_{1,1},
    \end{equation*}
    is the matrix multiplication tensor $\inprod{N,1,1}$.

    Moreover, let us consider an arbitrary non-empty block $A_t \times B_t \times C$. Tensor $T$ restricted to this block is
    \begin{equation*}
        T|_{A_t \times B_t \times C} = \sum_{j,k} x_{j,t} y_{t, k} z_{k,j},
    \end{equation*}
    which is exactly a $\inprod{1,N,N}$ matrix multiplication tensor.

    Note now that the projection $\pi_3 : A_t \times B_t \times C \to C$ induces a bijection between $C$ and $Y := \supp(T|_{A_t \times B_t \times C})$. If we look at $S := I \cap Y$, we wish to argue that $S$ is an independent set in $T|_{A_t \times B_t \times C}$. Indeed, denoting $S_i := \pi_i(S)$, since $(I, \mmP)$ is a permutation problem, all we need to argue is that 
    \begin{equation*}
        \pi_1^{-1}(S_1) \times \pi_2^{-1}(S_2) \times \pi_3^{-1}(S_3) \cap Y = S.
    \end{equation*}
    But since $\pi_3$ is bijection, we already have $(A_t \times B_t \times \pi_3^{-1}(S_3)) \cap Y = S$.
\end{proof}

We are now ready to prove \Cref{thm:laser-method-ub}.
\begin{proof}[Proof of \Cref{thm:laser-method-ub}]
For permutation NOF problems $I, J, K$, according to \Cref{lem:block-perm} we can chose a block-colored tensor structure $T_I, T_J, T_K$, such that $\Out(T_I) = \inprod{N,1,1}$, $\Out(T_J) = \inprod{1,N,1}$ and $\Out(T_K) = \inprod{1,1,N}$.

By \Cref{fct:block-product} and \Cref{fct:mat-mult-product}, the block-colored tensor $T = T_I \otimes T_J \otimes T_K$ satisfies $\Out(T) = \inprod{N,N,N}$, and by \Cref{fct:chromatic-numbers-product} every colored tensor $Z \in \Inn(T)$ has $\chi(Z) = 1$ (since it is a product of three colored tensors with chromatic number $1$).

By \Cref{cor:mat-mult-coloring}, $\chi(\inprod{N,N,N}) \leq N^{1 + o(1)}$, and finally using \Cref{lem:block-coloring}, 
\begin{equation*}
    \chi(T) \leq \chi(\Out(T)) \max_{T' \in \Inn(T)} \chi(T') \leq N^{1+o(1)}.
\end{equation*}
The communication complexity bound now follows from the characterization in \Cref{thm:coloring-is-communication}.
\end{proof}

\subsection{Zeroing out subrank of the matrix multiplication tensor \label{sec:mat-subrank}}
We will use the connection between number on the forehead communication complexity, and structural properties of matrix multiplication tensor to find a non-trivial \emph{upper bound} for the zero-out subrank of matrix multiplication tensor --- i.e., an upper bound on the size of the largest diagonal subtensor to which matrix multiplication tensor can be zeroed-out.

Authors of the work \cite{LPS17} observed a connection between the NOF communication complexity and the Ruzsa–Szemerédi problem. In particular, they proved a super-constant lower bound for NOF communication complexity of \emph{any} permutation problem, as a corollary of the Ruzsa–Szemerédi theorem (which in turns is proved via the triangle-removal lemma).

% Their lower bound, together with \cite{CFL83} protocol for $\eval(\bZ_N)$, is a roundabout route of proving Roth's theorem as a corollary of the Ruzsa–Szemerédi theorem --- if there was a constant-density subset of $\bZ_N$ avoiding three-terms arthmetic progressions, it could be used to design a $\Oh(1)$-communication protocol for $\eval_{\bZ_N}$, which contradicts the~\cite{LPS17} lower bound. (Much more direct way of deducing Roth's theorem from Ruzsa–Szemerédi theorem is well-known~\cite{blah}.)

Here, we observe that with our connection between NOF communication and the matrix-multiplication tensor, we can use their proof to show a somewhat stronger statement --- the upper bound on the zero-out subrank of matrix multiplication.

%The work \cite{LPS17} used triangle removal lemma to prove a super-constant lower bound for the communication complexity of all NOF permutation problems. On the other hand, essentially from the proof in \cite{CFL83}, one can deduce that existence of large sets without three-terms arthmetic progression in a group $G$, implies the existence of an efficient communication protocol for the $\mathrm{Eval}$ problem on this group (they proved this for the special case of the group 
%$\bZ_N$, but the proof is identical). As such the lower bound in \cite{LPS17} implies the Roth's theorem (that every subset of $\bZ_N$ without three-terms arthemtic progressions needs to have subconstant density).
%\jnote{This is not quite true --- the theorem they prove is about colorings not independent sets --- but the way they prove it is through independent sets.}
%Indeed, the proof of their lower bound can be understood as following closely the previously discovered proof of the Roth's theorem via triangle removal lemma. We observe that through our connection with matrix multiplication, again the same proof can be used to prove an upper bound on the zero-out subrank $Q_{zo}(\inprod{n,n,n})) \leq n^{2}/2^{c \log^* n}$. This is in contrast with known lower-bounds on border subrank, where we know that $\underline{Q}(\inprod{n,n,n}) = \lfloor \frac{3}{4} n^2 \rfloor$.

\begin{theorem}
\label{thm:subrank-triangle-lb}
The zero-out subrank of matrix multiplication satisfies
\begin{equation*}
    Q_{zo}(\inprod{N,N,N})) \leq N^{2}/2^{c \log^* N}.
\end{equation*}
\end{theorem}

Note that this theorem (together with \Cref{fct:qzo-lb}) implies in fact the $\Omega(\log^* N)$ lower bound for all NOF permutation problems present in~\cite{LPS17}, and the proof is almost identical.
%\joshnote{Does it follow from this that any NOF permutation problem requires communication $\Omega(\log^* n)$? Is this interesting? Is this what Toni already proved?}
%\jnote{Yes, this is interesting, and exactly what Toni proved. We can't use her result black-box, because we need upper bound on the independent set, and their result black-box is only lower-bound on chromatic number, but the identical proof applies here. I'll try to express exactly this sentiment (+ the fact that this is \emph{equivalent} to the Ruzsa-Szemeredi problem) better in this section.}

%Before we proceed with the theorem, we shall state a triangle removal lemma, and its consequence for our application
%\begin{theorem}[Triangle Removal Lemma]
%For every $\varepsilon$ there is $\delta$, such that if a graph $G$ has $\leq \delta n^3$ triangles, it can be made triangle-free by removing at most $\varepsilon n^2$ edges.

%Moreover $\delta^{-1}$ can be taken to be a tower function of height $\log \varepsilon^{-1}$.
%\end{theorem}

Before we continue, we shall state a convenient well-known equivalent formulation of the Ruzsa–Szemerédi theorem. The sub-quadratic upper bound of the quantity below has been famously proven by Ruzsa and Szemerédi in 1978~\cite{RS78}, using the notorious triangle removel lemma (in fact they considered a slightly different formulation, the $(6,3)$-problem, but it can be easily seen to be equivalent to the following statement~\cite{CEMS91}). Substituting a quantitative version of the triangle removal lemma by Fox~\cite{F11} in their proof leads to the following quantitative bound.  
%\jnote{As I just learned, this is exactly the Ruzsa–Szemerédi problem. Need to update intro to this section.}
\begin{theorem}[\cite{RS78,CEMS91,F11}]
\label{thm:ruzsa-szemeredi}
    For any graph $G(V,E)$, let $\cC$ be a collection of all triangles in the graph $G$. If all those triangles are edge-disjoint, then $|\cC| \leq N^2/2^{c \log^* N}$ where $N = |V|$.
\end{theorem}

With this theorem in hand, it is easy to show a sub-quadratic upper bound on the largest independent set in the matrix multiplication tensor. In fact the quantitative question about the density of the largest independent set in the matrix multiplication tensor is equivalent to the Ruzsa–Szemerédi problem.
\begin{proof}[Proof of \Cref{thm:subrank-triangle-lb}]
Consider an independent set $\tilde{S} \subset \inprod{N,N,N}$ --- we can treat it as a subset $S \subset [N] \times [N] \times [N]$. Let us now consider the following tri-partite graph on $[N] \times [3]$: for each $(a,b,c) \in S$ we put a triangle on vertices $(a,0), (b,1)$ and $(c,2)$. This construction yields a graph $G$ on $3 N$ vertices, together with a collection $\cC$ of its $|S|$ triangles. To appeal to the \Cref{thm:ruzsa-szemeredi}, we need to show that triangles in the collection $\cC$ are edge-disjoint, and that those are all the triangles in a graph $G$.

First of all, those are edge disjoint --- if we had a pair of distinct triangles in $\cC$ sharing an edge, it would correspond to two elements $(a,b, c_1) \in S$ and $(a,b, c_2) \in S$ --- a contradiction with an assumption that $S$ is an independent set in matrix multiplication tensor.

Similarly, since the graph $G$ is tri-partite, the only possible triangle are of form $(a,0), (b, 1), (c, 2)$ for some $a, b, c$. We wish to show that if $S$ was an independent set, this is possible only for $(a', b', c') \in S$.

Indeed, if the edge $(a, 0), (b, 1)$ was included in the graph, there must have been a tuple $(a, b, c') \in S$, and analogously $(a, b', c) \in S, (a',b,c) \in S$ for some $a', b', c'$. Since $S$ was an independent set, this readily implies that $a'=a, b'=b, c'=c$ and indeed $(a,b,c) \in S$.
% \jnote{Elaborate on this.}
\end{proof}

\subsubsection{Lower bound for matrix multiplication subrank from NOF protocols}

In the study of the matrix multiplication tensor, lower bounds on its zero-out subrank are much more consequential than upper bounds. Strassen proved that $\alpha(\inprod{N,N,N}) \geq N^{2-o(1)}$, which is almost maximal (since the tensor $\inprod{N,N,N}$ has size $N^2 \times N^2 \times N^2$, the largest diagonal tensor one could potentially hope to find there is of size $N^2$). He used this together with the Laser method to design a fast matrix multiplication algorithm in \cite{S86}, and all subsequent record-holding matrix multiplication algorithms have also used this. At the same time, recent barrier results \cite{AW21,A21,CVZ21} used the lower bounds on the subrank of matrix multiplication to prove barrier against certain approaches for proving that the matrix multiplication constant $\omega=2$. (Roughly speaking, using our notation here, they show that since $\alpha(\inprod{N,N,N})$ is so large, one would need to use an intermediate tensor $T$ with $\alpha(T)$ also large in order to design a sufficiently fast algorithm.)

One of the crucial ingredients in the proof using the Laser method that $\alpha(\inprod{N,N,N}) \geq N^{2 - o(1)}$ is leveraging the existence of a subset $S \subset [N]$ which does not have three-term arithmetic progressions, and has relatively high density. Using Behrend's \cite{behrend} construction of such a set with density $2^{-\Oh(\sqrt{\log N})}$, it follows that $\alpha(\inprod{N,N,N}) \geq N^2/2^{\Oh(\sqrt{\log  N})}$.

Using the same construction of Behrend's set, it was proved in \cite{CFL83} that the NOF problem $\eval(\bZ_N)$ has a deterministic protocol with communication complexity $\Oh(\sqrt{\log N})$ --- significantly improving upon the naive $\Oh(\log N)$. We observe that their construction of an efficient NOF protocol not only uses the same technical ingredients as the construction of large independent set in the matrix multiplication tensor, but in fact it is much more intimately related --- known results in communication complexity together with our connection directly imply the Strassen result on the subrank of matrix multiplication tensor.

\begin{proposition}
\label{fct:qzo-lb}
If there is any permutation problem $I \subset [N]\times [N] \times [N]$ with NOF communication complexity $\cc(\tilde{I}, \mmP) \leq k$, then $\alpha(\inprod{N,N,N}) \geq N^{2} 2^{-k}$.
\end{proposition}
\begin{proof}
    By \Cref{thm:coloring-is-communication}, if the promise problem $(\tilde{I}, \mmP)$ has NOF communication complexity $k$, then $\chi(I, \mmP) \leq 2^k$, and in particular taking the largest color we get $\alpha(I, \mmP) \geq n^{2} 2^{-k}$. Since the promise tensor $\mmP$ for the number on the forehead communication problems is exactly the matrix multiplication tensor $\inprod{n,n,n}$, we can disregard which terms are green, and deduce $\alpha(\inprod{n,n,n}) \geq \alpha(I, \mmP) \geq n^2 2^{-k}$.
\end{proof}

\subsection{Other promise problems}
As we understand now, the slice rank-based approach for proving lower bounds for asymptotic communication complexity of NOF permutation problems failed, because the technique depends only on the promise tensor, and is insensitive to the specific problem --- and the promise tensor corresponding to the number-on-the-forehead problems is a matrix multiplication tensor, with maximal subrank.

With this in mind, it is worth looking at promise problems with other promise-tensors, hopefully ones for which we know the asymptotic subrank is not maximal --- this statement alone shall imply that \emph{every} permutation problem with such a promise should have positive asymptotic communication complexity.

\begin{fact}
\label{fct:group-promise-hard}
    If the promise tensor $P$ over $\Sigma \times \Sigma \times \Sigma$ has non-maximal asymptotic subrank
    \begin{equation*}
        \underline{\alpha}(P) := \lim \sup \alpha(P^{\otimes N})^{1/N} < |\Sigma|,
    \end{equation*}
    then for every permutation problem $I \subset P$ we have 
    \begin{equation*}
        \cc((I, P)^{\otimes N}) \geq \Omega(N).
    \end{equation*}
\end{fact}
The proof of this fact is straightforward with what we have discussed so far. We later provide a complete proof of a strictly stronger theorem (\Cref{thm:asymptspectrum}), so we leave a direct proof of this special case to the reader. 

The solution of the cap-set problem \cite{CLP17,ED17,T16} and its subsequent generalizations \cite{BCHGNSU17}, can be interpreted as saying that for any abelian group $G$, the structure tensor
\begin{equation*}
    T_G := \sum_{g,h \in G} x_{g} y_{h} z_{-g-h},
\end{equation*}
has $\underline{\alpha}(T_G) < |G|$. In particular, all permutation problems with a promise tensor being a structure tensor of such a group, have positive asymptotic communication complexity. For example, the following is true
\begin{corollary}
    Consider the promise problem $\eq_{\F_3^N}$, where parties receive inputs $x,y,z \in \F_3^N$ with a promise that $x+y+z = 0$, and they wish to decide whether $x=y$. This problem requires $\Omega(N)$ communication.
\end{corollary}

Without the preceding discussion it would not be immediately clear whether we should expect this statement to be true --- the relevant part of the problem is between Alice and Bob, and they just wish to decide the equality predicate for their inputs --- a task which, with ordinary deterministic two-party protocols, requires $\Omega(N)$ communication. But should we expect Charlie, who knows $x+y$ to be of much help in solving the task?

On the other hand, one might think that the statement above --- if true --- should be relatively straightforward to prove directly, with a proof similar to the communication lower bound for deterministic two-party equality problem. We observe that in fact this statement is easily \emph{equivalent} to the cap-set theorem, hence a direct proof would be an interesting new proof of what used to be a cap-set \emph{conjecture} for almost four decades.

\begin{fact}
For any abelian group $G$, if there is a $S \subset G$ of size $|G|^{1-\delta}$ without three-terms arthmetic progressions, then $\cc(\eq_G) \leq \delta \log |G| + \Oh(\log \log |G|)$
\end{fact}
\begin{proof}
The corresponding communication tensor of a problem $\eq_G$ is $(I, T_G)$ where $T_G$ is the structure tensor of the group $G$, and $I = \{ (x,x, -2x) : x \in G \} \subset T_G$.

We will first show that $\alpha(I, T_G) \geq |S|$. We wish to show that if $S$ does not have three terms arithmetic progressions, then $S' := \{ (x,x, -2x) : x \in S \} \subset I$ is an independent set. Indeed, for the sake of contradiction, let us assume that $(x,y,-x-y)$ for some $x\not=y$ has $x \in \pi_1(S'), y \in \pi_2(S'), -x-y \in \pi_3(S')$. This means $x,y \in S$, and $-x-y = - 2z$ for some $z \in S$ --- i.e. $x,z,y$ is a three terms arithmetic progression.

To lift a large independent set to a small coloring, we will refer to \Cref{lem:coloring-from-symmetry} --- it is enough to show that there is an automorphism of $(I, T_G)$ that maps arbitrary $(x,x,-2x)$ to $(y,y,-2y)$. This is given by a triple of permutations $\sigma_1 = \sigma_2 : w \mapsto w - x + y$, and $\sigma_3 : w \mapsto w + 2y - 2x$. Finally, by \Cref{thm:coloring-is-communication} a coloring of the tensor $(I, T_G)$ into $|G|^{\delta} \log |G|$ colors yields a protocol with complexity $\delta \log |G|  + \log \log |G|$.
\end{proof}

\subsection{Lower Bounds from the Asymptotic Spectrum of Tensors}
The results in this section concern promise problems which are not necessarily permutation problems. We will first need a simple combinatorial description of the communication complexity of such a problem.

For any transcript of a protocol $\pi$, the set of inputs $S_{\pi} \subset \Sigma \times \Sigma \times \Sigma$ which result in such a transcript is a combinatorial cube $A_\pi \times B_\pi \times C_\pi$ for some $A_\pi, B_\pi, C_\pi \subset \Sigma$. If the protocol is correct, in every such cube either all terms from $P$ are contained in $I$, or they all are not in $I$. As such, a correct protocol of complexity $k$ induces a partition of the tensor $P$ into $2^k$ monochromatic combinatorial cubes (where a cube is monochromatic in the preceding sense). 

In particular, for any promise problems $(I, P)$ all terms from $I$ can be covered by $2^{\cc(I, P)}$ combinatorial cubes, each of which does not contain any term in $P - I$.

\begin{theorem} \label{thm:subadditive}
Let $r : \F^{\Sigma \times \Sigma \times \Sigma} \to \mathbb{R}_{\geq 0}$ be any function which is:
\begin{itemize}
    \item \emph{sub-additive}, meaning, if $A,B \in \F^{\Sigma \times \Sigma \times \Sigma}$ then $r(A+B) \leq r(A) + r(B)$, and
    \item \emph{monotone under zeroing outs}, meaning, if  $A,B \in \F^{\Sigma \times \Sigma \times \Sigma}$ and $B$ is a zeroing out of $A$, then $r(B) \leq r(A)$.
\end{itemize}
Suppose $P$ is any promise tensor over $\Sigma \times \Sigma \times \Sigma$, and $I \subset P$ is any problem (not necessarily a permutation problem). Then, $$\cc(I,P) \geq \log_2\left( \frac{r(I)}{r(P)} \right).$$
\end{theorem}

%\jnote{Since we aren't talking about permutation problems anymore, we need to state the relation between coloring and communication complexity. I think the exact relation is --- CC is equivalent to monochromatic rectangles partitions, so indeed --- showing lower bounds on partitionining green terms into monochromatic rectangles implies lower bound on CC (but is not equivalent).}
%\joshnote{Yes I agree. Also I may not have used the same notation/terminology as you.}

\begin{proof}
   Let $k = 2^{\cc(I,P)}$. By the preceding discussion, there are $k$ tensors $T_1, \ldots, T_k \in \F^{\Sigma \times \Sigma \times \Sigma}$ which are each zeroing outs of $P$, such that $T_1 + \cdots + T_k = I$. We thus have $$r(I) \leq \sum_{i=1}^k r(T_i) \leq k \cdot r(P),$$
    as desired.
\end{proof}

%\joshnote{Do I need `monotone under zeroing outs' in the theorem statement? Does it follow from subadditive? I feel like I'm missing something dumb...}
%\jnote{Hm, yeah, I think it follows from $r$ being non-negative and sub-additive...}
%\joshnote{The proof would have to use that it's a zeroing out, right?}
%\jnote{}
%\joshnote{Rank is definitely sub-additive: If you have a rank expression for A and a rank expression for B, then their sum is a rank expression for A+B.}
%\jnote{Right, of course.}
%\joshnote{Yeah I thought there must be some dumb proof, but it would have to use that it's a zeroing out (since it's not true for just any pair of tensors), and I don't see how to use that in conjunction with just these triangle inequality things, so maybe it really doesn't follow.}
%\jnote{Mhm. I went through the same journey, just 20 minutes later ;-). I agree.}
% \jnote{Any norm over $\mathbb{R}^{N^3}$ is sub-additive, but it's easy to find one which could grow under zero-ing out some coordinates.}

Although it is relatively simple, \Cref{thm:subadditive} is quite powerful since there are many well-studied examples of functions $r$ in algebraic complexity and tensor combinatorics which are sub-additive and monotone under zeroing outs. Tensor rank ($R$) is the most prominent example, but there are also other well-studied rank variants such as `slice rank'~\cite{T16} and `geometric rank'~\cite{KMZ20}. (Note that subrank ($Q$) is super-additive, not sub-additive.)

\paragraph{Number In Hand}
For one example, consider the all-1s tensor $P_{NIH} \in \F^{\Sigma \times \Sigma \times \Sigma}$, which is the promise tensor for the Number In Hand model. This very simple tensor has rank $R(P_{NIH}) = 1$. Thus, by \Cref{thm:subadditive}, any tensor $I$ with $R(I) \geq |\Sigma|^{\Omega(1)}$ has essentially maximal Number In Hand communication complexity $\Omega(\log |\Sigma|)$. Note that all tensors $I \in \F^{\Sigma \times \Sigma \times \Sigma}$ have $R(I) \geq |\Sigma|$ except for tensors which are somehow ``degenerate''; for instance, all `concise' tensors have this property~\cite{CHLVW18}.

This instantiation of \Cref{thm:subadditive}, where the promise is the all-ones tensor and $r$ is the rank of the problem tensor $I$, is akin to the standard lower bound in two-party deterministic communication complexity by the $\log \mathrm{rank}$ of the communication matrix. The celebrated log-rank conjecture stipulates that in the standard setting this statement can be weakly inverted, and in fact any two-party communication problem (with trivial promise) has a protocol with complexity $(\log \mathrm{rank}(I))^c$ for some constant $c$.

\paragraph{Number On Forehead}
For another example, consider the tensor $P_{NOF} \in \F^{\Sigma^2 \times \Sigma^2 \times \Sigma^2}$, which is the promise tensor for the Number On Forehead model. As we've seen, $P_{NOF} = \langle |\Sigma|, |\Sigma|, |\Sigma| \rangle$, and so we know that $R(P_{NOF}) = |\Sigma|^{\omega + o(1)}$. It follows that if $I \subset P_{NOF}$ is a problem with $R(I) \geq |\Sigma|^{c}$ for some constant $c > \omega$, then $I$ has Number On Forehead communication complexity $\geq (c - \omega + o(1)) \log |\Sigma|$. 

This relates the classic open problem in algebraic complexity of finding high-rank tensors to the task of proving deterministic NOF lower bounds. Furthermore, the fact that $\omega$, the exponent of matrix multiplication, is \emph{subtracted} in the communication lower bound is intriguing. This means that designing faster matrix multiplication algorithms, and hence decreasing our best-known upper bound on $\omega$, actually makes it \emph{easier} to prove NOF \emph{lower bounds} via this approach.

\subsubsection{Generalization to the Asymptotic Spectrum}

Finally, in this section, we note that while subrank is not sub-additive, the zeroing out variant ($Q_{ZO}$) we have been studying is additive (and hence sub-additive) for direct sums, and so a similar version of \Cref{thm:subadditive} still holds for it. Moreover, this holds for any $r$ in the `asymptotic spectrum of tensors'~\cite{S86,WZ22}, a well-studied class which roughly captures the different ways to generalize the notion of the `rank' of a matrix to tensors.

\begin{definition}
    Tensor $B$ is an \emph{identification} of tensor $A$ if you can get to $B$ from $A$ by first doing a zeroing out then setting some variables equal to each other. For example, $B = x_0 y_0 z_0 + x_0 y_1 z_1$ is an identification of $A = x_0 y_0 z_0 + x_1 y_1 z_1 + x_0 y_1 z_2$ by zeroing out $z_2$ and setting $x_1 \leftarrow x_0$. Identifications generalize zeroing outs, but are a special case of tensor restrictions.
\end{definition}

\begin{theorem} \label{thm:asymptspectrum}
Let $r : \F^{\Sigma \times \Sigma \times \Sigma} \to \mathbb{R}_{\geq 0}$ be any function which is:
\begin{itemize}
    \item \emph{sub-additive for direct sums}, meaning, if $A \in \F^{\Sigma \times \Sigma \times \Sigma}$ and $B \in \F^{\Sigma' \times \Sigma' \times \Sigma'}$ are tensors for disjoint sets $\Sigma, \Sigma'$, then $r(A+B) \geq r(A) + r(B)$, where $A+B$ is a tensor over $\F^{\Sigma \cup \Sigma' \times \Sigma \cup \Sigma' \times \Sigma \cup \Sigma'}$ and
    \item \emph{monotone under identifications}, meaning, if  $A,B \in \F^{\Sigma \times \Sigma \times \Sigma}$ and $B$ is an identification of $A$, then $r(B) \leq r(A)$.
\end{itemize}

Suppose $P$ is any promise tensor over $\Sigma \times \Sigma \times \Sigma$, and $I \subset P$ is any problem (not necessarily a permutation problem). Then, $$\cc(I,P) \geq \log_2\left( \frac{r(I)}{r(P)} \right).$$
\end{theorem}

\begin{proof}
   Let $k = 2^{\cc(I,P)}$. By definition, there are $k$ tensors $T_1, \ldots, T_k \in \F^{\Sigma \times \Sigma \times \Sigma}$ which are each zeroing outs of $P$, such that $T_1 + \cdots + T_k = I$. In particular, the direct sum of $k$ copies of $P$ has a zeroing out to the direct sum of all the $T_i$ tensors, which in turn has an identification to $I$. We thus have $$r(I) \leq \sum_{i=1}^k r(T_i) \leq k \cdot r(P),$$
    as desired.
\end{proof}

\begin{fact}
    The zeroing out subrank $r = Q_{ZO}$ satisfies the premises of \Cref{thm:asymptspectrum}.
\end{fact}

\begin{proof}
    $Q_{ZO}$ is sub-additive for direct sums since any zeroing out of a direct sum to a diagonal tensor must separately zero out each part into a diagonal tensor. It's monotone under identifications since, when transforming a tensor $A$ to a diagonal tensor via identification, we may assume that we first perform a zeroing out and then set variables equal to each other, but setting variables equal to each other cannot increase the size of a diagonal tensor.
\end{proof}

\subsection{Intermediate Group Promise Problems}

Suppose $T_1 \subset T_2 \subset T_3$ are $\{0, 1\}$-valued tensors. These give rise to three promise problems, depending on which one picks as the promise and which one picks as the problem: $(T_1, T_2)$, $(T_1, T_3)$, and $(T_2, T_3)$. We observe a simple triangle inequality-type relationship between their communication complexities:

\begin{fact} \label{fact:triangleineq}
    $\cc(T_1, T_3) \leq \cc(T_1, T_2) + \cc(T_2, T_3)$
\end{fact}

\begin{proof}
    In order to solve $(T_1, T_3)$, first use a protocol for $(T_2, T_3)$ to determine whether the input is in $T_2$, then use a protocol for $(T_1, T_2)$ to determine whether the input is in $T_1$.
\end{proof}

We will now consider a particular instantiation of \Cref{fact:triangleineq} which yields a potential avenue toward proving deterministic NOF lower bounds. $T_1, T_2, T_3$ will be tensors from $\F^{\Sigma^2 \times \Sigma^2 \times \Sigma^2}$ for some finite set $\Sigma$. Let $n := |\Sigma|$.

We will pick $T_3$ to be the structure tensor of the group $(\mathbb{Z} / n\mathbb{Z})^2$. This tensor is typically written
$$T_3 = \sum_{a,b,c,d \in [n]} x_{(a,b)} y_{(c,d)} z_{(a+c, b+d)},$$
where additions in the subscripts (here and in the remainder of this subsection) are done mod $n$. However, to simplify the next step, we can permute the variable names (replacing $x_{(a,b)}$ with $x_{(a+b,b)}$ for all $a,b \in [n]$, replacing $y_{(c,d)}$ with $y_{(c,c+d)}$ for all $c,d \in [n]$, and replacing $z_{(a+c, b+d)}$ with $z_{(b+d, a+c)}$ for all $a,b,c,d \in [n]$) to equivalently write it as
$$T_3 = \sum_{a,b,c,d \in [n]} x_{(a+b,b)} y_{(c,d+c)} z_{(b+d, a+c)}.$$

Next, we will pick $T_2 = \mmP \in \F^{n^2 \times n^2 \times n^2}$ to be the Number On Forehead promise tensor, i.e., the $\langle n,n,n \rangle$ matrix multiplication tensor,
$$T_2 = \sum_{i,j,k \in [n]} x_{(i,j)} y_{(j,k)} z_{(k,i)}.$$

\begin{fact}
    $T_2 \subset T_3$.
\end{fact}

\begin{proof}
For each $i,j,k \in [n]$, we need to prove that there is a choice of $a,b,c,d \in [n]$ such that $x_{(i,j)} y_{(j,k)} z_{(k,i)} = x_{(a+b,b)} y_{(c,d+c)} z_{(b+d, a+c)}$. The choice which achieves this is $b=c=j$, $a = i-j \pmod{n}$, and $d=k-j \pmod{n}$.
\end{proof}

Finally, we may pick $T_1$ to be any Number On Forehead problem $T_1 \subset T_2$. \Cref{fact:triangleineq} now relates the Number On Forehead complexity of the problem $T_1$ to two different promise problems with the promise $T_3$:

$$\cc(T_1, P_{NOF}) \geq \cc(T_1, T_3) - \cc(P_{NOF}, T_3).$$

Using the tools we've developed thusfar, we can bound the communication complexity $\cc(P_{NOF}, T_3)$. Note that these are tensors over $\F^{n^2 \times n^2 \times n^2}$, so the trivial communication upper bound would be $2 \log n$.

%\jnote{Even if $\omega=2$ we should have some lower bound $c \log n \leq \cc(\mmP, T_3)$, just by taking a problem for which $\cc(T_1, \mmP)$ is small (like $\eval_{\bZ_N}$), and $\cc(T_1, T_3)$ is big (by the slice-rank \Cref{fct:group-promise-hard})? Let's try to work out the constant.}
%\joshnote{As $n$ gets big, $T_3$ approaches maximal slice rank. Like I think the cap set statement is that for every $n$ there is a $c_n > 0$ such that $T_3$ has asymptotic slice rank at most $(1 - c_n) \cdot n$. But $c_n \to 0$ as $n \to \infty$.}
%\jnote{But since $\inprod{2,2,2} \subset T_{\bZ_2^2}$, can't we just take a large power to deduce $\inprod{2^n,2^n,2^n} \subset $ }
%\joshnote{Yeah but that's a different tensor from this $T_3$. Like it's for the group $(\mathbb{Z}/2\mathbb{Z})^n$ rather than $\mathbb{Z}/2^n\mathbb{Z}$}
%\jnote{Ok, I see.}
%\joshnote{I think this one's better since it doesn't have that lower bound. Maybe we should talk about both though.}
%\jnote{Yeah, I agree, it's better. We might write about it, I wasn't sure how this compares with out conversation, now I see.}
\begin{lemma} \label{lem:NOFingroup}
    $$(\omega - 2) \log n \leq \cc(P_{NOF}, T_3) \leq \log(n).$$
\end{lemma}

\begin{proof}
    Since $P_{NOF}$ is the $\langle n,n,n \rangle$ matrix multiplication tensor, its rank is $R(P_{NOF}) = n^{\omega + o(1)}$. Meanwhile, since $T_3$ is the structure tensor of an abelian group, its rank is $R(T_3) = n^2$ (see, e.g.,~\cite[{Theorem 2.3 and Theorem 4.1}]{CU03}). Thus, \Cref{thm:subadditive} yields the lower bound $(\omega - 2) \log n \leq \cc(P_{NOF}, T_3)$.
    
    For the upper bound, we directly give a communication protocol. In this problem, player $A$ is given $(a+b,b)$, player $B$ is given $(c,d+c)$, and player $C$ is given $(b+d,a+c)$, for some values $a,b,c,d \in \mathbb{Z}_n$, and their goal is to determine whether or not the following three equalities hold: $b=c$, $d+c=b+d$, and $a+c=a+b$. Note that this is equivalent to testing whether $b=c$, since the other two equalities follow from this. They can do this by having player $A$ send $b$, which uses $\log n$ bits, and then having player $B$ confirm that it is equal to $c$.
\end{proof}

Interestingly, in light of the lower bound in \Cref{lem:NOFingroup}, we know that improving the upper bound $\cc(P_{NOF}, T_3) \leq \log(n)$ \emph{requires} using fast matrix multiplication. Indeed, any bound $\cc(P_{NOF}, T_3) \leq (1 - \delta) \log(n)$ would imply that $\omega < 3 - \delta$. More specifically, one can observe that an upper bound on $\omega$ obtained in this way would fall under the `group-theoretic approach' to matrix multiplication~\cite{CU03}, using the group $(\mathbb{Z} / n\mathbb{Z})^2$. There are known barriers to this approach for any fixed $n$~\cite{BCHGNSU17}, but it's consistent with the best-known barriers that this approach could prove $\omega = 2$ by using increasingly larger $n$. (On the other hand, most recent work on the group-theoretic approach has focused instead on non-abelian groups~\cite{BCGPU22}.)

Combining \Cref{lem:NOFingroup} together with \Cref{fact:triangleineq} shows that
\begin{fact}\label{fact:noftogroup}
    $$\cc(T_1, P_{NOF}) \geq \cc(T_1, T_3) - \log(n).$$
\end{fact}
For any $T_1$, the straightforward algorithm for the problem $(T_1, T_3)$ has communication complexity $2 \log (n)$. \Cref{fact:noftogroup} shows that any $T_1$ with $\cc(T_1, T_3) \geq (1 + \delta) \log (n)$ for any fixed $\delta>0$ would give a linear Number On Forehead lower bound $\cc(T_1, P_{NOF}) \geq \Omega( \log n)$.

In light of this, improving the upper bound of \Cref{lem:NOFingroup} could be quite powerful. In particular, if $\omega = 2$, and if this can be used to give an algorithm achieving $\cc(P_{NOF}, T_3) \leq o(\log(n))$, then any $T_1$ with $\cc(T_1, T_3) \geq \Omega(\log (n))$ (without any restriction on the leading constant) would give a linear Number On Forehead lower bound $\cc(T_1, P_{NOF}) \geq \Omega( \log n)$ as well.

\paragraph{Acknowledgements}
The authors thank Madhu Sudan and Toniann Pitassi for helpful discussions.
%\subsection{Papers we should look at and probably discuss}
%\jnote{I just ran onto \cite{CFTZ22}, where they define \emph{symmetric subrank} and prove some stuff about it. I think parts of this research used to be the part of the previous paper \cite{CFTZ21}, but they split it into two papers before submitting to ITCS}

%\jnote{Also: we need to understand the entire mess about Ruzsa-Szemeredi graphs and connections with what we did. Noga Alon has a very recent paper on this \cite{AS20}, they cite Ankur Moitra, and Toni \cite{AMS12,LPS17}.}
%\jnote{This work proves that some particular tensor has maximal asymptotic subrank \cite{AVZ19}. They might be using laser method, I'm not sure. It's unclear why the tensor is important, but the paper is doing some really heavy lifting to arrive at this result.}
\bibliographystyle{alpha}
\bibliography{biblio}

\end{document}